\documentclass[10pt,twocolumn, nofootinbib]{revtex4-2}

\usepackage{assumptionsofphysics}
\usepackage{graphicx}
\usepackage{hyperref}
\hypersetup{
	colorlinks=true,
	citecolor=blue,
	urlcolor=blue,
	linkcolor=blue
}
\begin{document}

\title{Reverse Physics: From Laws to Physical Assumptions}
\author{Gabriele Carcassi, Christine A. Aidala}
\affiliation{Physics Department, University of Michigan, Ann Arbor, MI 48109}

\date{\today}

\begin{abstract}
To answer foundational questions in physics, physicists turn more and more to abstract advanced mathematics, even though its physical significance may not be immediately clear. What if we started to borrow ideas and approaches, with appropriate modifications, from the foundations of mathematics? In this paper we explore this route. In reverse mathematics\cite{friedman1976systems,simpson2017reverse,stillwellreverse} one starts from theorems and finds the minimum set of axioms required for their derivation. In reverse physics we want to start from laws or more specific results, and find the physical concepts and starting points that recover them. We want to understand what physical results are implied by which physical assumptions. As an example of the technique, we will see six different characterizations of classical mechanics, show that the uncertainty principle depends only on the entropy bound on pure states and recast the third law of thermodynamics in terms of the entropy of an empty system. We believe the approach can provide greater insights into both current and new physical theories, put the physical concepts at the forefront of the discussion and provide a more unified view of physics by highlighting common patterns and ideas across different physical theories.
\end{abstract}

\maketitle

\section{Introduction}

Within the foundations of mathematics, the reverse mathematics program\cite{friedman1976systems,simpson2017reverse,stillwellreverse} aims to start from known theorems and determine what axioms are required to prove them. This helps understand the relative logical strength of the theorems and get a better sense of their relationships. We propose a similar program for physics. We start from known physical theories or elements of physical theories, and we try to find premises that enable us to rederive them. The goal is to find premises that are physically more intuitive, so that ideally the arguments can be carried out with or without the math, and develop concepts that are of physical significance across theories.

While one does find efforts that analyze the premises of physical theories (see e.g. Refs.~\cite{chiribella2011informational, selby2021reconstructing, giles2016mathematical, boyling1972axiomatic, PhysRev.101.860, Haag1964848, axiomaticQFT1975, masanes2019measurement, carcassi2021four}), for example in the context of foundations of quantum mechanics, statistical mechanics or the search for quantum gravity, their aim is typically different. Some are focused on a particular field (e.g. quantum theory) or try to demonstrate the centrality of one aspect (e.g. information) with regard to others. Some concentrate more on mathematical premises or structure, making the physical significance less than clear. Many aim to develop new theories, instead of first gaining further insight within current ones. We propose a more generalist approach that examines the relationships of physical ideas across subfields, which is best explained through examples.

To give a sense of the breadth of the reverse physics approach, we show how it can be used to pursue three different objectives. First, we illustrate multiple instances of \emph{reformulation}. We study Hamilton's equations for a single degree of freedom to find six alternative characterizations that will connect different ideas from different fields under the same theme of determinism and reversibility. Next, we illustrate \emph{dependence analysis}. We show that the uncertainty principle is not a consequence of the full quantum theory, but simply of its zero entropy bound on pure states. If a similar bound is imposed in classical mechanics, similar uncertainty relationships are recovered. Lastly, we illustrate \emph{reconceptualization}. We take the third law of thermodynamics and rework it so that the role of crystalline substances to define zero entropy is instead played by the ``empty system''.

As we want the discussion to be focused on the physics, we will use the most widely known math among physicists and relegate mathematical derivations to the appendix. We also want to stress that this paper is not about the originality of each single result, but about how these are used to address a more general and higher goal: to uncover the core physical concepts and assumptions in the known physical theories and understand their generative power and their interrelationships.

\section{Determinism, reversibility and Hamiltonian mechanics}

Let us start our discussion with the equations of motion given by classical Hamiltonian mechanics:
\begin{align}\label{HamiltonEquations1}
	\frac{dq}{dt} = \frac{\partial H}{\partial p} \; \; \; \; \; \; \frac{dp}{dt} = - \frac{\partial H}{\partial q}
\end{align}
Our first task is to find a set of mathematical conditions that are equivalent to these equations.

Let us group the state variables $\xi^a = \{q, p\}$ and consider the displacement vector field $S = \frac{d\xi^a}{dt} = \{ \frac{dq}{dt}, \frac{dp}{dt} \} $ which tells the direction in phase space in which each state moves over time. We find that $S$ is divergenceless:
\begin{align}\label{divergenceless}
\frac{\partial S^q}{\partial q} + \frac{\partial S^p}{\partial p} =0
\end{align}
meaning that the flow through any closed region in phase space is zero. We also find the converse is true: a two dimensional divergenceless field $S$ always allows a stream function $H$ such that
\begin{align}
	S^q = \frac{\partial H}{\partial p} \; \; \; \; \; \; S^p = - \frac{\partial H}{\partial q}
\end{align}
In other words, Hamiltonian evolution for one degree of freedom is exactly the evolution for which the flow over a closed region in phase space is zero.

Let us now look at how an infinitesimal region of phase space evolves. Suppose $Q = q + \frac{dq}{dt} dt = q + S^q dt$ and $P = p + \frac{dp}{dt} dt = p + S^p dt$. The infinitesimal region of area $dq dp$ will become $dQ dP = |J| dq dp$ where:
\begin{align}
	|J| = \begin{vmatrix}
		\frac{\partial Q}{\partial q} & \frac{\partial Q}{\partial p} \\
		\frac{\partial P}{\partial q} & \frac{\partial P}{\partial  p} 
	\end{vmatrix} = \frac{\partial Q}{\partial q} \frac{\partial P}{\partial  p} - \frac{\partial P}{\partial q} \frac{\partial Q}{\partial p}
\end{align}
Note how the Jacobian determinant, in this simple case, coincides with the Poisson bracket. With a simple substitution we find:
\begin{align}
	|J|\approx 1 + \left( \frac{\partial S^q}{\partial q} + \frac{\partial S^p}{\partial p} \right)dt.
\end{align}
Note that the first order term is exactly the divergence, therefore condition \eqref{divergenceless} is equivalent to
\begin{align}\label{UnitaryJacobian}
	|J| = 1
\end{align}
which means the initial area $dqdp$ is equal to the final area $dQdP$.

So we have found that, for a two dimensional manifold, Hamiltonian evolution \eqref{HamiltonEquations1}, zero flow over closed regions \eqref{divergenceless} and preservation of area \eqref{UnitaryJacobian} are mathematically the same condition.\footnote{This is essentially a short proof for Liouville's theorem that can work in both directions.} However, we still have to answer the main question: what does this mean physically?

To do that, we look at statistical mechanics. There we use the area of a phase space region to count the number of states. Area preservation, then, means that the evolution preserves the state count: we start and we end with the same number of states. We have a one-to-one map between past and future states. Given an initial state we have one and only one final state. Physically, this means
\begin{align}\label{DeterminismReversibility}
	\text{Deterministic and reversible evolution}
\end{align}
is equivalent to conditions \eqref{HamiltonEquations1}, \eqref{divergenceless} and \eqref{UnitaryJacobian}. But determinism and reversibility is a physical condition, a physical property of the evolution: we have a physical characterization of Hamiltonian evolution.

This seems to tell us that the physical requirement of deterministic and reversible evolution over the continuum mathematically is not simply a bijection. The reason is that, over a continuum, counting points is not enough. A segment a meter long has as many points as a segment a kilometer long, but they are not the same length. In the same way, all finite regions of phase space have infinitely many points, but physically they do not contain the same number of states. Mathematically, we need a measure to give a size to each region, and therefore deterministic and reversible evolution is a bijection that preserves the measure, that preserves the state count: Hamiltonian evolution.

Note that this makes sense from a statistical mechanics/thermodynamics perspective as well. In statistical mechanics, the entropy is the logarithm of the count of states, therefore conservation of number of states means conservation of entropy. This means a deterministic evolution that is thermodynamically reversible (conservation of entropy) is also reversible in the dynamical sense (conservation of state count). Therefore 
\begin{align}\label{DeterminismThermodynamicReversibility}
	\parbox{2.8in}{Deterministic and thermodynamically reversible evolution}
\end{align}
is yet another equivalent condition.

What about entropy in the information sense? If the evolution is deterministic and reversible, the amount of information needed to specify the initial state should be exactly the same as the amount of information needed to specify the final state. That is, if $I[\rho] = - \int \rho \log \rho dqdp$ is the Gibbs/Shannon entropy, we would expect the following condition:
\begin{align}\label{EntropyConservation}
	I[\rho(t_1)] = I[\rho(t_2)]
\end{align}
to be equivalent to the others. We find:
\begin{align}
	I[\rho(t + dt)] = I[\rho(t)] - \int \rho \log |J| dqdp.
\end{align}
Since the Jacobian determinant of a continuous transformation cannot be negative, condition \eqref{EntropyConservation} is indeed equivalent to \eqref{UnitaryJacobian} and therefore to all the others.

It may seem odd that conservation of entropy for a deterministic system gives us energy conservation: we typically associate energy conservation with system isolation in thermodynamics. Where is the connection? For a deterministic and reversible system, future and past states depend only on the state of the present system. Therefore they do not depend on the state of other systems or of the environment. This means that deterministic and reversible systems are necessarily isolated. This time, running the argument in the opposite direction presents a problem. One needs to first decide whether a non-deterministic isolated system makes sense. If one assumes that the source of non-determinism, of stochastic uncertainty, is always interaction with other systems, other degrees of freedom, then every isolated system is deterministic. In this case, we can run the argument in reverse: every isolated system is deterministic and reversible. Without that extra assumption, this reverse direction is not guaranteed.

Finally, let us go back to the idea that a bijection is not enough to characterize deterministic and reversible evolution. Preservation of a measure is a mathematical idea. Is there a more physical way to look at it? The issue is that scientific measurements over a continuum can only carry finite precision. When considering how the information goes back and forth in time, then, we also need to take into account how the finite precision is mapped. While it may be true that the evolution of a damped harmonic oscillator is a bijection, the points get denser and denser around the equilibrium. Once we are close to equilibrium it becomes impossible to tell when the oscillator was started: the damped harmonic oscillator is not reversible in any practical sense if precision is taken into consideration. Let us see if this argument can be run formally. We want a coordinate invariant quantity that characterizes uncertainty. The obvious choice is the determinant of the covariance matrix:
\begin{align}\label{CovarianceMatrix}
	|\Sigma| = \begin{vmatrix}
		\sigma_q^2 & cov_{q,p} \\
		cov_{p,q} & \sigma_p^2
	\end{vmatrix} = \sigma_q^2 \sigma_p^2 - cov_{q,p}^2
\end{align}
We can therefore imagine the following condition
\begin{align}\label{BijectionUncertainty}
	|\Sigma(t_1)| = |\Sigma(t_2)| 
\end{align}
for which the uncertainty remains constant in time. Is this yet another equivalent condition? If we assume the spread of the distribution is small enough, we find
\begin{align}
	|\Sigma(t + dt)| = |J| |\Sigma(t)| |J|.
\end{align}
therefore \eqref{BijectionUncertainty} is indeed equivalent to \eqref{UnitaryJacobian} and to all others.

We could go on to find other relationships and extend the ones we found to multiple degrees of freedom, but we believe this should be enough to highlight the power of the reverse physics approach. The first thing to note is how in a couple of pages we have found fundamental connections between classical Hamiltonian mechanics \eqref{HamiltonEquations1}, vector calculus \eqref{divergenceless}, differential geometry \eqref{UnitaryJacobian}, statistical mechanics \eqref{DeterminismReversibility}, thermodynamics \eqref{DeterminismThermodynamicReversibility}, information theory \eqref{EntropyConservation} and plain statistics \eqref{BijectionUncertainty}. This helps foster a sense of unity of these disparate disciplines and their perspectives, a sense that is sorely lacking both in research and education: nature is one, and does not care about how we have divided academic knowledge. We believe that a single unified view of this kind will bring more coherence and clarity to physics than, for example, a unified theory for the fundamental forces.

Another advantage is that we were able to carry out many of the arguments conceptually. Mathematics is then used to better articulate the physical arguments. This should please those who believe that physics, not mathematics, should be more at the center of the discussion.\cite{hossenfelder2018lost, woit2006not}

The other interesting aspect is how much we were able to find in theories that are generally considered well understood. While it is common for some people to be aware of some of these results, we have found that most results are unknown to most people.

\section{Uncertainty principle revisited}

Let us now turn our attention to quantum mechanics, and see if our approach can shed new light on a theory that is not generally considered to be well understood. Since we talked about the role of uncertainty in classical mechanics, let us concentrate on the uncertainty principle, which states that every state has to satisfy the relationship
\begin{align}\label{UncertaintyPrinciple}
	\sigma_q \sigma_p \geq \frac{\hbar}{2}.
\end{align}

We can start with an interesting observation. If we look back at \eqref{CovarianceMatrix}, we have:
\begin{align}
	\sigma_q^2 \sigma_p^2 = |\Sigma| + cov_{q,p}^2 \geq |\Sigma|.
\end{align}
Since by \eqref{BijectionUncertainty} $|\Sigma|$ is a constant of motion, we find that during Hamiltonian evolution the uncertainty is bounded. The lowest uncertainty is reached when there are no correlations and the covariance is zero. Though it is a conceptually different relation, since $|\Sigma|$ is just a constant of motion, it gives us the following intuition: a coordinate invariant cap to the uncertainty will produce an inequality on the product of variances. Where can we find such a cap in quantum mechanics?

Let us turn our attention to entropy, represented by the von Neumann entropy $I[\rho] = - tr (\rho \log \rho)$. For every pure state $|\psi \rangle$ we have the following property:
\begin{align}\label{ZeroQuantumEntropy}
	\begin{split}
		I[|\psi\rangle\langle\psi|] = 0.
	\end{split}
\end{align}
In other words, all pure states have zero entropy, the entropy is capped. Since entropy and uncertainty are related, is fixing the amount of entropy enough to recover the uncertainty principle?

We can test that hypothesis by studying the space of classical distributions with a fixed value of entropy $I_0$. We find that they have to satisfy the relationship\footnote{The careful reader will note the units do not quite work. The issue here is the $\log \rho$ in the Shannon/Gibbs entropy, as $\rho$ is not a pure number. Introducing a dimensionful constant $\log \rho \hat{h}$ would fix the expression, which would then fix the uncertainty relationship as well. }
\begin{align}\label{ClassicalUncertaintyPrinciple}
	\sigma_q \sigma_p \geq \frac{e^{I_0}}{2 \pi e}.
\end{align}
The equality holds if the distribution is the product of two independent Gaussians.

The following condition:
\begin{align}\label{BoundedEntropy}
	\text{The entropy of the system is fixed}
\end{align}
is therefore enough to recover an uncertainty relationship independently of the theory. This can be understood as fixing the amount of information carried by the state. The reverse argument does not work, so condition \eqref{BoundedEntropy} is a stronger condition than the existence of an uncertainty relationship.

This example shows another advantage of reverse physics: it allows us to pinpoint what part of a theory is responsible for which effect. In this case, we have found that it is not quantum mechanics per se that leads to uncertainty relationships, just the entropy cap on pure states. If we implemented a similar cap in classical mechanics, we would obtain the same result. Therefore condition \eqref{ZeroQuantumEntropy} is a more fundamental and clear way to characterize the uncertainty in quantum mechanics, making it evident that all pure states carry the same information, they describe the system at the same level of precision.

\section{A new understanding of the third law}

In the previous section we saw the lower entropy bound of zero entropy is built into quantum mechanics and not in classical mechanics. However, the bound is also built into thermodynamics through the third law\footnote{This formulation is given by Ref. \cite{lewis1923thermodynamics}}
\begin{align}\label{ThirdLaw}
	\parbox{2.8in}{Every substance has a finite positive entropy, but
		at the absolute zero of temperature the entropy may become zero, and does so
		become in the case of perfect crystalline substances.}
\end{align}
Therefore we have two fundamental physical theories that share this trait. Can we find a more general characterization for this?

The first step is to look for a system that can conceptually function as a zero for entropy that feels a bit less arbitrary than a perfect crystalline substance at zero temperature. Let us recall that a fundamental property for thermodynamic entropy is that it is an extensive quantity, it is additive for independent systems:
\begin{align}
	I_{AB} = I_{A} + I_{B}.
\end{align}
Is there a system that acts as a zero for system composition? This would be the empty system $\emptyset$: any system combined with the empty system $\emptyset$ will remain unchanged.\footnote{If we consider systems as a monoid under composition, the empty system is the identity element, much like the number zero, the empty set or the identity map in their respective structures.} In particular, we must have:
\begin{align}\label{EntropyEmptySystem}
	I_{\emptyset} = I_{\emptyset\emptyset} = I_{\emptyset} + I_{\emptyset} = 0.
\end{align}
We can therefore use the empty system as our zero reference for entropy.

This move is a conceptual shift. In thermodynamics, the idea that the entropy for crystalline substances is zero is a phenomenological assumption: it is so because it justifies the behavior of substances as they approach absolute zero. We need statistical mechanics to explain it. Condition $\eqref{EntropyEmptySystem}$, however, is not a phenomenological consideration but a conceptual necessity: it must be so. In this sense, the new zero entropy system is more fundamental. But does it relate to the old ideas? In terms of statistical mechanics, the crystalline substance has zero entropy because it has one possible way to be. The empty system also has one possible configuration. So the old ideas indeed carry over: a crystalline substance will have to have the same entropy as the empty system.

Now that we have a better concept for zero entropy, we need to understand why we can't have states with lower entropy. Let us consider it from an information theoretical perspective. Suppose we had a state with negative entropy. This would mean that it is better specified than a state with zero entropy. The state is better specified than saying that the system is not there. But stating that the system is not there already completely specifies the system. Therefore the condition that imposes a lower bound of zero for entropy is the following:
\begin{align}\label{NewThirdLaw}
	\parbox{2.8in}{No state can describe a system more accurately than stating the system is not there in the first place.}
\end{align}
Again, note that this condition is not phenomenological: it is a logical necessity. This condition is realized in quantum mechanics: pure states represent the most precise descriptions of the system and they have zero entropy much like the vacuum, another pure state. This condition is realized in thermodynamics through the third law. In fact, the third law could be rederived as follows. Assuming \eqref{NewThirdLaw} tells us that entropy is bounded at zero. Since the entropy in thermodynamics is a convex function of energy, entropy will reach its minimum at zero temperature, which is the only temperature that can reach zero entropy. We then use the previous result that the entropy of a crystalline substance at zero temperature is the same as for an empty system.

This showcases the final advantage of reverse physics. We want to differentiate between ``assumptions'' (conditions that are not necessarily true) and ``principles'' (those that have to be taken to be true); we want to differentiate between conditions that are ``phenomenological'' (empirically derived or justified) from those that are ``conceptual'' (logical necessities or consequences of definitions).\footnote{Philosophers may call these ``constitutive'' conditions.} If we are able to elevate assumptions to principles, or replace phenomenological assumptions with conceptual ones, then the premises for our theories are sturdier and we have improved their foundations.

\section{Conclusion}

We have shown how the reverse physics approach, trying to find different premises from which to rederive physical theories or specific results, can help us improve our understanding. In the first example we showed how we can take a single theory (classical Hamiltonian mechanics) and characterize it in several ways consistent with each other (deterministic evolution that is either dynamically or thermodynamically reversible, conservation of information entropy, conservation of measurement precision, ...). We have seen how this helps foster a unified understanding of physics, and puts physics (and not mathematics) at the center of the discussion. We then showed how we can isolate a specific result of a theory (the uncertainty relationship), see what specific part of the theory is responsible for that result (fixing the entropy) and showed how another theory (classical mechanics) can be modified to achieve the same result. Lastly, we saw how we can recast a mostly phenomenological condition (the zero entropy of a crystalline structure at zero temperature) into a conceptual one (the zero entropy of an empty system).

We believe this type of work can be beneficial in several ways. At the very least, it can have a significant impact within physics education. Most importantly, it leads to common ideas and definitions that span different fields of physics, which may help develop common tools, much like set theory and category theory provide standard tools for all of mathematics. Lastly, as it helps us understand the realm of validity of the different theories and also the underlying mathematical tools, it may give insight into the development of new tools and new theories.

More than the particular results presented in this paper, it is the type of approach that we want to promote. We hope others will present their own reverse physics results, and that this may lead to a renewed interest in the manner in which physicists in the past tried to formulate the foundations of our field: through physical principles, laws and assumptions.

\section*{Acknowledgments}
We acknowledge funding from the MCubed program of the University of Michigan. This work is part of a larger project, Assumptions of Physics~\cite{aop-book}, which aims to identify a handful of physical principles from which the basic laws can be rigorously derived.

\section*{Appendix: Proofs and calculations}

We include here all proofs which we omitted from the main body since, as we stated, the physics should be given full attention. To keep the discussion light and accessible to the widest audience, we use the simplest mathematical techniques to get the result. Some readers may feel that some results are well known, or too simple to even be included. Due to the interdisciplinary nature of the subject, we believe it is useful to collect them all here.

\begin{prop}
	A dynamical system characterized by two variables $\{q, p\}$ follows Hamilton's equations if and only if the displacement field $S=\{\frac{dq}{dt}, \frac{dp}{dt} \}$ is divergenceless.
\end{prop}

\begin{proof}
	Suppose the evolution is Hamiltonian. Then $S=\{\frac{\partial H}{\partial p}, - \frac{\partial H}{\partial q} \}$. The divergence is $div(S) = \frac{\partial S^q}{\partial q} + \frac{\partial S^p}{\partial p} = \frac{\partial^2 H}{\partial q\partial p} - \frac{\partial^2H}{\partial p\partial q} = 0$.
	
	Conversely, suppose $S$ is divergenceless. Then, it admits a stream function $H$ such that $\{\frac{dq}{dt}, \frac{dp}{dt}\} = \{\frac{\partial H}{\partial p}, - \frac{\partial H}{\partial q} \}$, which are Hamilton's equations.
\end{proof}

\begin{prop}
	The displacement field $S=\{\frac{dq}{dt}, \frac{dp}{dt} \}$ of a dynamical system characterized by two variables $\{q, p\}$ is divergenceless if and only if the Jacobian determinant of the infinitesimal time evolution is equal to one.
\end{prop}

\begin{proof}
	Let $Q = q + \frac{dq}{dt} dt = q + S^q dt$ and $P = p + \frac{dp}{dt} dt = p + S^p dt$. The Jacobian determinant is given by:
	\begin{align*}
		|J| &= \begin{vmatrix}
			\frac{\partial Q}{\partial q} & \frac{\partial Q}{\partial p} \\
			\frac{\partial P}{\partial q} & \frac{\partial P}{\partial  p} 
		\end{vmatrix} = \frac{\partial Q}{\partial q} \frac{\partial P}{\partial  p} - \frac{\partial P}{\partial q} \frac{\partial Q}{\partial p} \\
		&= \left(1 + \frac{\partial S^q}{\partial  q}dt\right)\left(1 + \frac{\partial S^p}{\partial  p}dt\right) - \left(\frac{\partial S^p}{\partial  q}dt\right)\left(\frac{\partial S^q}{\partial  p}dt\right) \\
		&= 1 + \left(\frac{\partial S^q}{\partial  q} + \frac{\partial S^p}{\partial  p}\right)dt + O(dt^2)
	\end{align*}
	The Jacobian determinant will be 1 if and only if the first order term is zero. Since the first order term is the divergence of $S$, this proves the proposition.	
\end{proof}

\begin{remark}
	An unstated assumption here is that $H$ is twice differentiable. The Jacobian and the divergence, in fact, would not be well defined if $H$ were differentiable only once. In that case, the acceleration $a = \frac{d}{dt}\frac{dq}{dt} = \frac{d}{dt} \frac{\partial H}{\partial p}$ would also be ill defined. With this in mind, if a Hamiltonian is not twice differentiable at a point it is better to consider the system not Hamiltonian at that point.
\end{remark}

\begin{prop}
	The information entropy is conserved during a continuous transformation if and only if the Jacobian is unitary.
\end{prop}

\begin{proof}
	Let $I[\rho(x^i)] = - \int \rho(x^i) \log \rho(x^i) \, dx^n$ be the Shannon entropy of the distribution $\rho(x^i)$. Let $y^j=y^j(x^i)$ be a differentiable transformation. Since $\rho(x^i)$ is a density, we have $\rho(y^j) = \rho(x^i) |J|$. The information entropy after the transformation will be
	\begin{align*}
		I[\rho(y^j)] &= - \int \rho(y^j) \log \rho(y^j) \, dy^n \\
		&= - \int \rho(x^i) |J| \log \left( \rho(x^i) |J| \right) \, dy^n \\
		&= - \int \rho(x^i) \log \left( \rho(x^i) |J| \right) \, dx^n \\
		&= - \int \rho(x^i) \log \rho(x^i) \, dx^n - \int \rho(x^i) \log |J| \, dx^n \\
		&= I[\rho(x^i)] - \int \rho(x^i) \log |J| \, dx^n.
	\end{align*}
	The information entropy is conserved for every $\rho$ if and only if $|J|$ is 1 everywhere. Therefore the entropy is conserved if and only if the Jacobian is unitary.
\end{proof}

\begin{prop}
	A continuous transformation will conserve the uncertainty of a distribution with small support if and only if the Jacobian is unitary. The uncertainty is characterized by the determinant of the covariance matrix and small means the evolution can be considered approximately linear over the support of the distribution.
\end{prop}

\begin{proof}
	As we assume the distribution to be small enough, we can linearize the transformation to $y^j = J^j_i x^i + B^j$ where $J$ is the Jacobian. Noting that the covariance is a linear operator that does not depend on the expectation of the variables, we have
	\begin{align*}
		|\Sigma^{ij}| &= |cov(y^i, y^j)| \\
		&= |cov(J^i_k x^k + B^i, J^j_l x^l + B^j)| \\
		&= |cov(J^i_k x^k , J^j_l x^l )| \\
		&= |J^i_k cov(x^k , x^l ) J^j_l| \\
		&= |J^i_k | | \Sigma^{kl} | | J^j_l| \\
	\end{align*}
	The uncertainty is conserved for all small distributions if and only if $|J|$ is $\pm 1$ everywhere. Given that the transformation is continuous, the Jacobian determinant must be positive (i.e. a continuous transformation cannot be a reflection) and therefore the entropy is conserved if and only if the Jacobian is unitary.
\end{proof}

\begin{prop}
	The von Neumann entropy of a pure state is zero.
\end{prop}

\begin{proof}
	Let $\rho$ be the density matrix of a pure state. We have $\rho = |\psi \rangle \langle \psi | = |\psi \rangle \langle \psi |\psi \rangle \langle \psi | = \rho^2$. The von Neumann entropy is $I[\rho] = - tr (\rho \log \rho) = - tr (\rho \log \rho^2) = - 2 tr (\rho \log \rho) = 2 I[\rho]$. Therefore $I[\rho] = 0$.
\end{proof}

\begin{prop}
	Let $\rho(q,p)$ be a normalized density distribution over the 2-dimensional manifold charted by $(q,p)$. Furthermore, let $I_0$ be the value of the Shannon/Gibbs entropy $I[\rho]$. Then
	$$ 	\sigma_q \sigma_p \geq \frac{e^{I_0}}{2 \pi e}.$$
	Furthermore, the equal sign applies in the case where $\rho$ is the product of two gaussians.
\end{prop}

\begin{proof}
	We set up a minimization problem using Lagrange multipliers. We want to minimize the product of the variance $\sigma_q^2 \sigma_p^2 \equiv \int (q-\mu_q)^2 \rho_{\mathcal{c}} \, dqdp \int (p-\mu_p)^2 \rho_{\mathcal{c}} \, dqdp$ while keeping $\rho$ normalized and fixing the entropy to $I_0$. We have: 
	\begin{align*}
		L = &\int (q-\mu_q)^2 \rho_{\mathcal{c}} \, dqdp \int (p-\mu_p)^2 \rho_{\mathcal{c}} \, dqdp \\
		&+ \lambda_1(\int \rho_{\mathcal{c}} dqdp - 1) \\ &+ \lambda_2(- \int \rho_{\mathcal{c}} \ln \rho_{\mathcal{c}} \, dqdp - I_0)\\
		\delta L = &\int \delta \rho_{\mathcal{c}} [(q-\mu_q)^2 \sigma_p^2 + \sigma_q^2 (p-\mu_p)^2 + \\ &\lambda_1 - \lambda_2 \ln \rho_{\mathcal{c}} - \lambda_2 ] dqdp = 0 \\
		\lambda_2 \ln \rho_{\mathcal{c}} = &\lambda_1 - \lambda_2 + (q-\mu_q)^2 \sigma_p^2 + \sigma_q^2 (p-\mu_p)^2 \\
		\rho_{\mathcal{c}} = &e^{\frac{\lambda_1 - \lambda_2}{\lambda_2}}e^{\frac{(q-\mu_q)^2 \sigma_p^2}{\lambda_2}}e^{\frac{\sigma_q^2 (p-\mu_p)^2}{\lambda_2}}
	\end{align*}
	We solve the multipliers and have:
	\begin{align*}
		\rho_{\mathcal{c}} = &\frac{1}{ 2 \pi \sigma_q \sigma_p} e^{-\frac{(q-\mu_q)^2}{2\sigma_q^2}} e^{-\frac{(p-\mu_p)^2}{2\sigma_p^2}} \\
		I_0 = &\ln (2\pi e\sigma_q\sigma_p) \\
		\sigma_q \sigma_p &= \frac{e^{I_0}}{2 \pi e}
	\end{align*}
	This shows that the gaussian minimizes the spread at fixed entropy, therefore all other distributions must have a larger or equal spread. 
\end{proof}

\begin{remark}
	Note that the inverse argument does not work: if we have a bound on the uncertainty, we cannot say anything about the entropy. All distributions with higher entropy will satisfy a higher bound, therefore the entropy is arbitrarily high. Moreover, we can find distributions with low entropy but with arbitrarily high spread. Therefore a bound on the uncertainty still allows any value for entropy.
\end{remark}

\bibliography{bibliography}

\begin{thebibliography}{16}%
\makeatletter
\providecommand \@ifxundefined [1]{%
 \@ifx{#1\undefined}
}%
\providecommand \@ifnum [1]{%
 \ifnum #1\expandafter \@firstoftwo
 \else \expandafter \@secondoftwo
 \fi
}%
\providecommand \@ifx [1]{%
 \ifx #1\expandafter \@firstoftwo
 \else \expandafter \@secondoftwo
 \fi
}%
\providecommand \natexlab [1]{#1}%
\providecommand \enquote  [1]{``#1''}%
\providecommand \bibnamefont  [1]{#1}%
\providecommand \bibfnamefont [1]{#1}%
\providecommand \citenamefont [1]{#1}%
\providecommand \href@noop [0]{\@secondoftwo}%
\providecommand \href [0]{\begingroup \@sanitize@url \@href}%
\providecommand \@href[1]{\@@startlink{#1}\@@href}%
\providecommand \@@href[1]{\endgroup#1\@@endlink}%
\providecommand \@sanitize@url [0]{\catcode `\\12\catcode `\$12\catcode
  `\&12\catcode `\#12\catcode `\^12\catcode `\_12\catcode `\%12\relax}%
\providecommand \@@startlink[1]{}%
\providecommand \@@endlink[0]{}%
\providecommand \url  [0]{\begingroup\@sanitize@url \@url }%
\providecommand \@url [1]{\endgroup\@href {#1}{\urlprefix }}%
\providecommand \urlprefix  [0]{URL }%
\providecommand \Eprint [0]{\href }%
\providecommand \doibase [0]{https://doi.org/}%
\providecommand \selectlanguage [0]{\@gobble}%
\providecommand \bibinfo  [0]{\@secondoftwo}%
\providecommand \bibfield  [0]{\@secondoftwo}%
\providecommand \translation [1]{[#1]}%
\providecommand \BibitemOpen [0]{}%
\providecommand \bibitemStop [0]{}%
\providecommand \bibitemNoStop [0]{.\EOS\space}%
\providecommand \EOS [0]{\spacefactor3000\relax}%
\providecommand \BibitemShut  [1]{\csname bibitem#1\endcsname}%
\let\auto@bib@innerbib\@empty
\bibitem [{\citenamefont {Friedman}(1976)}]{friedman1976systems}%
  \BibitemOpen
  \bibfield  {author} {\bibinfo {author} {\bibfnamefont {H.~M.}\ \bibnamefont
  {Friedman}},\ }\bibfield  {title} {\bibinfo {title} {Systems on second order
  arithmetic with restricted induction i, ii},\ }\href@noop {} {\bibfield
  {journal} {\bibinfo  {journal} {J. Symb. Logic}\ }\textbf {\bibinfo {volume}
  {41}},\ \bibinfo {pages} {557} (\bibinfo {year} {1976})}\BibitemShut
  {NoStop}%
\bibitem [{\citenamefont {Simpson}(2017)}]{simpson2017reverse}%
  \BibitemOpen
  \bibfield  {author} {\bibinfo {author} {\bibfnamefont {S.~G.}\ \bibnamefont
  {Simpson}},\ }\href@noop {} {\emph {\bibinfo {title} {Reverse mathematics
  2001}}},\ Vol.~\bibinfo {volume} {21}\ (\bibinfo  {publisher} {Cambridge
  University Press},\ \bibinfo {year} {2017})\BibitemShut {NoStop}%
\bibitem [{\citenamefont {Stillwell}(2019)}]{stillwellreverse}%
  \BibitemOpen
  \bibfield  {author} {\bibinfo {author} {\bibfnamefont {J.}~\bibnamefont
  {Stillwell}},\ }\href@noop {} {\emph {\bibinfo {title} {Reverse Mathematics:
  Proofs from the Inside Out}}}\ (\bibinfo  {publisher} {Princeton University
  Press},\ \bibinfo {year} {2019})\BibitemShut {NoStop}%
\bibitem [{\citenamefont {Chiribella}\ \emph {et~al.}(2011)\citenamefont
  {Chiribella}, \citenamefont {D’Ariano},\ and\ \citenamefont
  {Perinotti}}]{chiribella2011informational}%
  \BibitemOpen
  \bibfield  {author} {\bibinfo {author} {\bibfnamefont {G.}~\bibnamefont
  {Chiribella}}, \bibinfo {author} {\bibfnamefont {G.~M.}\ \bibnamefont
  {D’Ariano}},\ and\ \bibinfo {author} {\bibfnamefont {P.}~\bibnamefont
  {Perinotti}},\ }\bibfield  {title} {\bibinfo {title} {Informational
  derivation of quantum theory},\ }\href@noop {} {\bibfield  {journal}
  {\bibinfo  {journal} {Physical Review A}\ }\textbf {\bibinfo {volume} {84}},\
  \bibinfo {pages} {012311} (\bibinfo {year} {2011})}\BibitemShut {NoStop}%
\bibitem [{\citenamefont {Selby}\ \emph {et~al.}(2021)\citenamefont {Selby},
  \citenamefont {Scandolo},\ and\ \citenamefont
  {Coecke}}]{selby2021reconstructing}%
  \BibitemOpen
  \bibfield  {author} {\bibinfo {author} {\bibfnamefont {J.~H.}\ \bibnamefont
  {Selby}}, \bibinfo {author} {\bibfnamefont {C.~M.}\ \bibnamefont
  {Scandolo}},\ and\ \bibinfo {author} {\bibfnamefont {B.}~\bibnamefont
  {Coecke}},\ }\bibfield  {title} {\bibinfo {title} {Reconstructing quantum
  theory from diagrammatic postulates},\ }\href@noop {} {\bibfield  {journal}
  {\bibinfo  {journal} {Quantum}\ }\textbf {\bibinfo {volume} {5}},\ \bibinfo
  {pages} {445} (\bibinfo {year} {2021})}\BibitemShut {NoStop}%
\bibitem [{\citenamefont {Giles}(2016)}]{giles2016mathematical}%
  \BibitemOpen
  \bibfield  {author} {\bibinfo {author} {\bibfnamefont {R.}~\bibnamefont
  {Giles}},\ }\href@noop {} {\emph {\bibinfo {title} {Mathematical foundations
  of thermodynamics: International series of monographs on pure and applied
  mathematics}}},\ Vol.~\bibinfo {volume} {53}\ (\bibinfo  {publisher}
  {Elsevier},\ \bibinfo {year} {2016})\BibitemShut {NoStop}%
\bibitem [{\citenamefont {Boyling}(1972)}]{boyling1972axiomatic}%
  \BibitemOpen
  \bibfield  {author} {\bibinfo {author} {\bibfnamefont {J.}~\bibnamefont
  {Boyling}},\ }\bibfield  {title} {\bibinfo {title} {An axiomatic approach to
  classical thermodynamics},\ }\href@noop {} {\bibfield  {journal} {\bibinfo
  {journal} {Proceedings of the Royal Society of London. A. Mathematical and
  Physical Sciences}\ }\textbf {\bibinfo {volume} {329}},\ \bibinfo {pages}
  {35} (\bibinfo {year} {1972})}\BibitemShut {NoStop}%
\bibitem [{\citenamefont {Wightman}(1956)}]{PhysRev.101.860}%
  \BibitemOpen
  \bibfield  {author} {\bibinfo {author} {\bibfnamefont {A.~S.}\ \bibnamefont
  {Wightman}},\ }\bibfield  {title} {\bibinfo {title} {Quantum field theory in
  terms of vacuum expectation values},\ }\href
  {https://doi.org/10.1103/PhysRev.101.860} {\bibfield  {journal} {\bibinfo
  {journal} {Phys. Rev.}\ }\textbf {\bibinfo {volume} {101}},\ \bibinfo {pages}
  {860} (\bibinfo {year} {1956})}\BibitemShut {NoStop}%
\bibitem [{\citenamefont {Haag}\ and\ \citenamefont
  {Kastler}(1964)}]{Haag1964848}%
  \BibitemOpen
  \bibfield  {author} {\bibinfo {author} {\bibfnamefont {R.}~\bibnamefont
  {Haag}}\ and\ \bibinfo {author} {\bibfnamefont {D.}~\bibnamefont {Kastler}},\
  }\bibfield  {title} {\bibinfo {title} {An algebraic approach to quantum field
  theory},\ }\href {https://doi.org/10.1063/1.1704187} {\bibfield  {journal}
  {\bibinfo  {journal} {Journal of Mathematical Physics}\ }\textbf {\bibinfo
  {volume} {5}},\ \bibinfo {pages} {848} (\bibinfo {year} {1964})}\BibitemShut
  {NoStop}%
\bibitem [{\citenamefont {Streater}(1975)}]{axiomaticQFT1975}%
  \BibitemOpen
  \bibfield  {author} {\bibinfo {author} {\bibfnamefont {R.~F.}\ \bibnamefont
  {Streater}},\ }\bibfield  {title} {\bibinfo {title} {Outline of axiomatic
  relativistic quantum field theory},\ }\href
  {https://doi.org/10.1088/0034-4885/38/7/001} {\ \textbf {\bibinfo {volume}
  {38}},\ \bibinfo {pages} {771} (\bibinfo {year} {1975})}\BibitemShut
  {NoStop}%
\bibitem [{\citenamefont {Masanes}\ \emph {et~al.}(2019)\citenamefont
  {Masanes}, \citenamefont {Galley},\ and\ \citenamefont
  {M{\"u}ller}}]{masanes2019measurement}%
  \BibitemOpen
  \bibfield  {author} {\bibinfo {author} {\bibfnamefont {L.}~\bibnamefont
  {Masanes}}, \bibinfo {author} {\bibfnamefont {T.~D.}\ \bibnamefont
  {Galley}},\ and\ \bibinfo {author} {\bibfnamefont {M.~P.}\ \bibnamefont
  {M{\"u}ller}},\ }\bibfield  {title} {\bibinfo {title} {The measurement
  postulates of quantum mechanics are operationally redundant},\ }\href@noop {}
  {\bibfield  {journal} {\bibinfo  {journal} {Nature communications}\ }\textbf
  {\bibinfo {volume} {10}},\ \bibinfo {pages} {1} (\bibinfo {year}
  {2019})}\BibitemShut {NoStop}%
\bibitem [{\citenamefont {Carcassi}\ \emph {et~al.}(2021)\citenamefont
  {Carcassi}, \citenamefont {Maccone},\ and\ \citenamefont
  {Aidala}}]{carcassi2021four}%
  \BibitemOpen
  \bibfield  {author} {\bibinfo {author} {\bibfnamefont {G.}~\bibnamefont
  {Carcassi}}, \bibinfo {author} {\bibfnamefont {L.}~\bibnamefont {Maccone}},\
  and\ \bibinfo {author} {\bibfnamefont {C.~A.}\ \bibnamefont {Aidala}},\
  }\bibfield  {title} {\bibinfo {title} {Four postulates of quantum mechanics
  are three},\ }\href@noop {} {\bibfield  {journal} {\bibinfo  {journal}
  {Physical Review Letters}\ }\textbf {\bibinfo {volume} {126}},\ \bibinfo
  {pages} {110402} (\bibinfo {year} {2021})}\BibitemShut {NoStop}%
\bibitem [{\citenamefont {Hossenfelder}(2018)}]{hossenfelder2018lost}%
  \BibitemOpen
  \bibfield  {author} {\bibinfo {author} {\bibfnamefont {S.}~\bibnamefont
  {Hossenfelder}},\ }\href@noop {} {\emph {\bibinfo {title} {Lost in math: How
  beauty leads physics astray}}}\ (\bibinfo  {publisher} {Hachette UK},\
  \bibinfo {year} {2018})\BibitemShut {NoStop}%
\bibitem [{\citenamefont {Woit}(2006)}]{woit2006not}%
  \BibitemOpen
  \bibfield  {author} {\bibinfo {author} {\bibfnamefont {P.}~\bibnamefont
  {Woit}},\ }\href@noop {} {\emph {\bibinfo {title} {Not even wrong: The
  failure of string theory and the search for unity in physical law}}}\
  (\bibinfo  {publisher} {Basic Books (AZ)},\ \bibinfo {year}
  {2006})\BibitemShut {NoStop}%
\bibitem [{\citenamefont {Lewis}\ and\ \citenamefont
  {Randall}(1923)}]{lewis1923thermodynamics}%
  \BibitemOpen
  \bibfield  {author} {\bibinfo {author} {\bibfnamefont {G.~N.}\ \bibnamefont
  {Lewis}}\ and\ \bibinfo {author} {\bibfnamefont {M.}~\bibnamefont
  {Randall}},\ }\href@noop {} {\emph {\bibinfo {title} {Thermodynamics and the
  free energy of chemical substances}}}\ (\bibinfo  {publisher} {McGraw-Hill},\
  \bibinfo {year} {1923})\BibitemShut {NoStop}%
\bibitem [{\citenamefont {Carcassi}\ and\ \citenamefont
  {Aidala}(2021)}]{aop-book}%
  \BibitemOpen
  \bibfield  {author} {\bibinfo {author} {\bibfnamefont {G.}~\bibnamefont
  {Carcassi}}\ and\ \bibinfo {author} {\bibfnamefont {C.~A.}\ \bibnamefont
  {Aidala}},\ }\href {https://doi.org/10.3998/mpub.12204707} {\emph {\bibinfo
  {title} {Assumptions of Physics}}}\ (\bibinfo  {publisher} {Michigan
  Publishing},\ \bibinfo {year} {2021})\BibitemShut {NoStop}%
\end{thebibliography}%

\end{document}